\documentclass[final]{cvpr}

\usepackage{times}
\usepackage{epsfig}
\usepackage{graphicx}
\usepackage{amsmath}
\usepackage{amssymb}

\usepackage{amsthm}
\usepackage{mathtools}
\usepackage{stfloats}
\usepackage{multirow}
\usepackage{subfigure} 
\usepackage{amsopn}
\usepackage{bm}
\usepackage{array}
\usepackage{booktabs}
\usepackage{enumerate}
\usepackage{color}
\usepackage{multirow}
\usepackage{adjustbox}

\DeclareMathOperator{\sign}{sgn}

\usepackage[pagebackref=true,breaklinks=true,colorlinks,bookmarks=false]{hyperref}

\def\cvprPaperID{1565} 
\def\confYear{CVPR 2021}

\begin{document}

\title{AdderSR: Towards Energy Efficient Image Super-Resolution}

\author{Dehua Song$^{1*}$, Yunhe Wang$^{1}$\thanks{Equal contribution}, Hanting Chen$^{1,2}$, Chang Xu$^{3}$, Chunjing Xu$^{1}$, Dacheng Tao$^{3}$ \\
\normalsize$^1$ Noah's Ark Lab, Huawei Technologies  \ \ \normalsize$^2$ Peking University \ \ \
\normalsize$^3$ The University of Sydney.\\
\small\texttt{\{dehua.song, yunhe.wang\}@huawei.com;\ \ \ dacheng.tao@sydney.edu.au;}
}

\maketitle

\begin{abstract}
This paper studies the single image super-resolution problem using adder neural networks (AdderNets). Compared with convolutional neural networks, AdderNets utilize additions to calculate the output features thus avoid massive energy consumptions of conventional multiplications. However, it is very hard to directly inherit the existing success of AdderNets on large-scale image classification to the image super-resolution task due to the different calculation paradigm. Specifically, the adder operation cannot easily learn the identity mapping, which is essential for image processing tasks. In addition, the functionality of high-pass filters cannot be ensured by AdderNets. To this end, we thoroughly analyze the relationship between an adder operation and the identity mapping and insert shortcuts to enhance the performance of SR models using adder networks. Then, we develop a learnable power activation for adjusting the feature distribution and refining details. Experiments conducted on several benchmark models and datasets demonstrate that, our image super-resolution models using AdderNets can achieve comparable performance and visual quality to that of their CNN baselines with an about 2.5$\times$ reduction on the energy consumption. The codes are available at: \rm https://github.com/huawei-noah/AdderNet.
\end{abstract}

\section{Introduction}

Single image super-resolution (SISR) is a typical computer vision task which aims at reconstructing a high-resolution (HR) image from a low-resolution (LR) image. SISR is a very popular image signal processing task in real-world applications such as smart phones and mobile cameras. Due to the hardware constrains of these portable devices, it is necessary to develop SISR models with low computation cost and high visual quality. 

Recently, deep convolutional neural network (CNN) has dramatically boosted the performance of SISR. The first super-resolution convolutional neural network (SRCNN)~\cite{dong2014learning} contains only three convolutional layers with about 57K parameters. Then, the capacity of DCNN was amplified with the increasing of depth and width (channel number), resulting in notable improvement of super-resolution. The parameters and computation cost of recent DCNN are increased accordingly. For example, the residual dense network (RDN)~\cite{zhang2018residual} contains 22M parameters and requires about 10,192G FLOPs (floating-number operations) for processing only one image. Compared with neural networks for visual recognition (\eg, ResNet-50~\cite{he2016deep}), models for SISR have much higher computational complexities due to the larger feature map sizes. These massive calculations will consume much energy and reduce the enduration time of mobile devices.

\begin{figure}[tp]
\setlength{\abovecaptionskip}{-0.2cm}
\begin{center}
\scalebox{1.0}{
\includegraphics[width=0.48\textwidth]{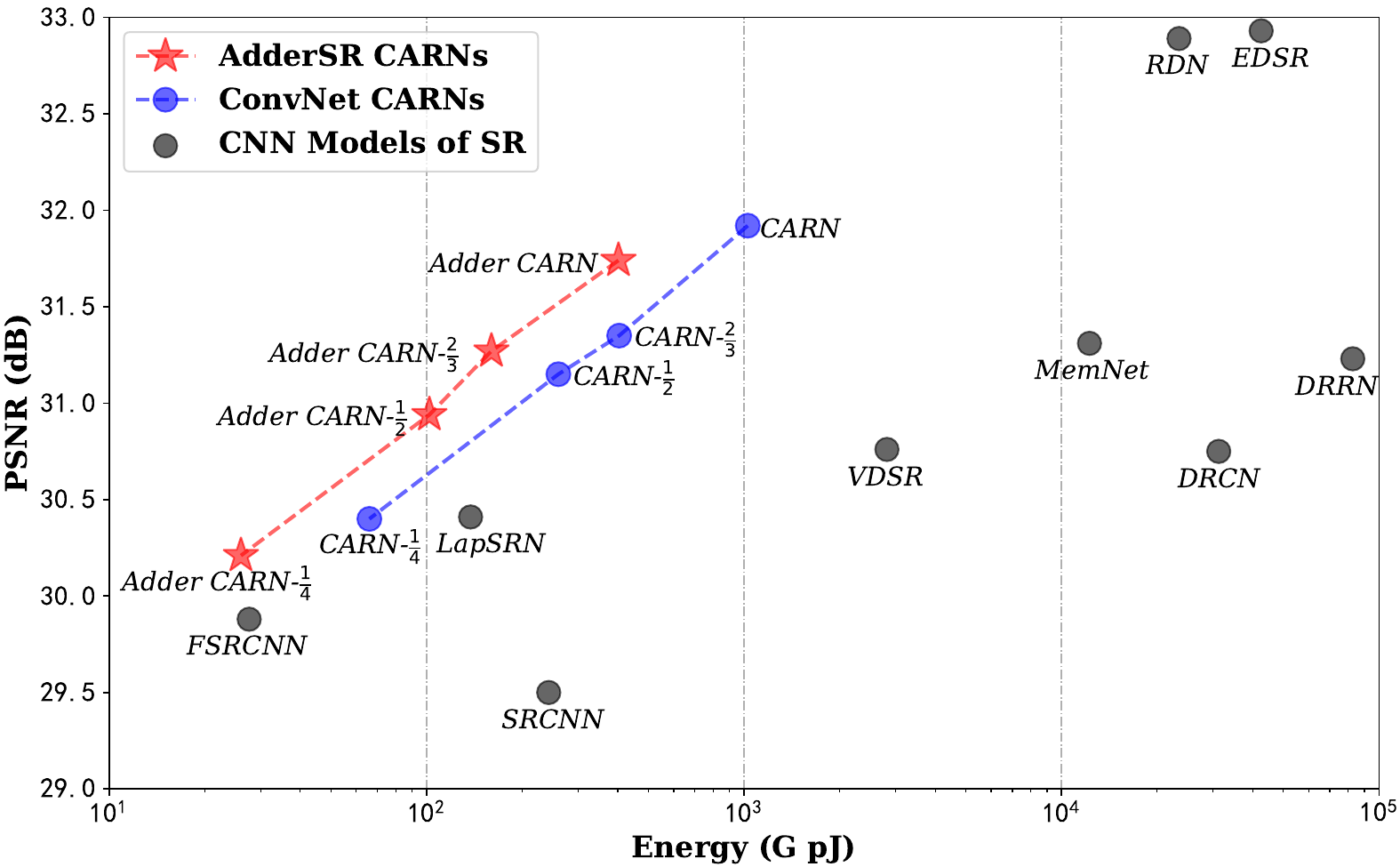}
}
\end{center}
\vspace{-0.3cm}
\caption{PSNR and energy trade-off with the SOTA SR methods on Urban100 database for $\times 2$ scale. AdderSR networks achieve superior performance with moderate energy.}
\label{FigEnergyCurve}
\vspace{-0.2cm}
\end{figure}

In order to address the aforementioned problem, a series of approaches have been proposed to compress and accelerate deep convolutional neural networks. The prominent compression methods such as filter pruning ~\cite{guo2016dynamic,wang2020gan,hou2020efficient,lin2020hrank} and knowledge distillation~\cite{hinton2015distilling,gao2018image,fu2020autogan,xu2020kernel} reduce the computation by narrowing or shallowing the network. On the other hand, quantization methods~\cite{courbariaux2015binaryconnect,ma2019efficient,Li2020PAMS,xin2020binarized} devote to reducing the computation complexity of multiplications while preserving the architecture of the original neural network. Wherein, binarization is a specific case that weights and activations in networks are represented as $\{+1, -1\}$, which can significantly reduce the energy and memory consumptions. However, the binarized network often cannot maintain the accuracy of full precision network, especially for super-resolution task~\cite{ma2019efficient}. Recently, Chen~\etal~\cite{chen2019addernet} proposed a novel AdderNet which replaces the multiplication operations by additions. Since the complexity of additions is much lower than that of multiplications, this work motivates us to utilize AdderNet for constructing energy efficient SISR models.

To maximally excavate the potential for exploiting AdderNets to establish SISR models, we first analyze the theoretical difficulties for applying the additions into SISR tasks. Specifically, input and output features in any two neighbor layers in SISR models are very close with the similar global texture and color information as shown in Figure~\ref{FigFeaturemaps}. However, the identify mapping cannot be learned by a one-layer adder network. Thus, we suggest to insert self-shortcuts and formulate new adder models for the SISR task. Moreover, we find that the high-pass filter is also hard to approximate by adder units. We then develop a learnable power activation. By exploiting these two techniques, we replace the conventional convolution filters in modern SISR networks by adder filters and establish AdderSR models accordingly. The effectiveness of the proposed SISR networks using additions is verified on several benchmark datasets. We can obtain comparable performance (\ie, PSNR values and visual quality) using AdderSR models with that of the CNN baselines. Meanwhile, we can reduce more than $50\%$ of the overall energy consumptions of these neural networks. Fig.~\ref{FigEnergyCurve} indicates the superior performance of AdderSR networks with moderate energy.

The rest of this paper is organized as follows. We briefly investigate the related works on neural network compression and energy-efficient approaches  in Section~\ref{RelatedwWorks}. In Section~\ref{AdderNetSR}, we present the motivation of using additions in SISR and establish AdderSR models. Section~\ref{experiment} illustrates both the quantitative and qualitative results on benchmarks and Section~\ref{conclusion} concludes the paper.

\section{Related Works}
\label{RelatedwWorks}

SISR is an important computer vision task, which has broad applications such as photo capture, surveillance and entertainment. In the last decades, numerous prominent approaches have been proposed to solve this ill-posed problem and obtained a tremendous progress in the performance. Unfortunately, current SISR methods require massive computations to process an input image. The energy consumption has become a thorny issue restricting the application of SISR on mobile devices.

\subsection{Model Compression}

Model compression has been investigated for many years and vast novel methods~\cite{han2020ghostnet,dong2016accelerating,kim2016deeply,tai2017image,lai2017deep,tai2017memnet} have been proposed. These methods can be roughly divided into four categories: network pruning, efficient filter design, neural architecture search (NAS) and knowledge distillation. Pruning~\cite{guo2016dynamic,wang2020gan,hou2020efficient,lin2020hrank} aims at reducing the redundancy of filters so as to decrease the computation of the original model. Hou~\etal~\cite{hou2020efficient} proposed a new pruning criterion for SISR which judged redundant channels with the discriminant information. The most conventional method is to design efficient block (\eg~GhostNet~\cite{han2020ghostnet}, CARN~\cite{ahn2018fast}, IDN~\cite{hui2018fast}, MAFFSRN~\cite{muqeet2020ultra}) with efficient operators (\eg group convolution, $1 \times 1$ convolution). Furthermore, NAS~\cite{song2019efficient, gong2019autogan, fu2020autogan} has been employed to exploit efficient SR neural architecture automatically. In addition, knowledge distillation~\cite{hinton2015distilling,gao2018image,fu2020autogan,xu2020kernel} can transfer the information from large models to improve the performance of the tiny model. Gao~\etal~\cite{gao2018image} proposed a novel knowledge distillation scheme for SISR and boosted the performance.

\subsection{Low-Cost Computation} 

In addition to reduce the computation of models, there are two kinds of approaches that decrease the  energy consumption~\cite{horowitz20141} while preserving the architecture of original network. Model quantization~\cite{courbariaux2015binaryconnect,ma2019efficient,Li2020PAMS,xin2020binarized} saves energy~\cite{horowitz20141} by reducing the number of bits required to represent each weight or feature component. Wherein, binarization is a specific case that weights and activations in networks are represented as $\{+1, -1\}$. Li \etal~\cite{Li2020PAMS} proposed a novel quantization scheme for SISR to acquire a large dynamic quantization range. Xin \etal~\cite{xin2020binarized} designed a bit-accumulation mechanism to alleviate the quantization error of binary SISR networks. Unfortunately, the quantized network often cannot maintain the accuracy of super-resolution networks. On the other hand, many research works~\cite{horowitz20141,sze2017efficient,you2020shiftaddnet} have indicated that add addition consumes fewer energy than multiplication. Recently, Chen et al.~\cite{chen2019addernet} pioneered a novel method to reduce the power dissipation of network by taking the place of the multiplication with add operation. It achieved marginal loss of accuracy on classification tasks without any multiplication in convolutional layers. Then, kernel based progressive distillation~\cite{xu2020kernel} promoted the accuracy of AdderNets even superior to that of standard CNNs. It is attractive to construct energy efficient SISR models with AdderNet.
\section{AdderNet for Image Super-Resolution}
\label{AdderNetSR}

AdderNet reduces the energy consumption of classification networks significantly while achieving comparable performance. We aim to inherit this huge success to the image super-resolution task, which often has higher energy consumption and computational complexity. In SISR, there are two important properties that should be ensured by adder neural networks: the similarity between the input and output features of each convolutional layer, and the enhancement of details \wrt high frequency information. These two properties are illustrated in Fig.~\ref{FigFeaturemaps}. 

\subsection{Preliminaries and Motivation}

\begin{figure*}[tp]
\setlength{\abovecaptionskip}{-0.2cm}
\begin{center}
\scalebox{1.0}{
\includegraphics[width=1.0\textwidth]{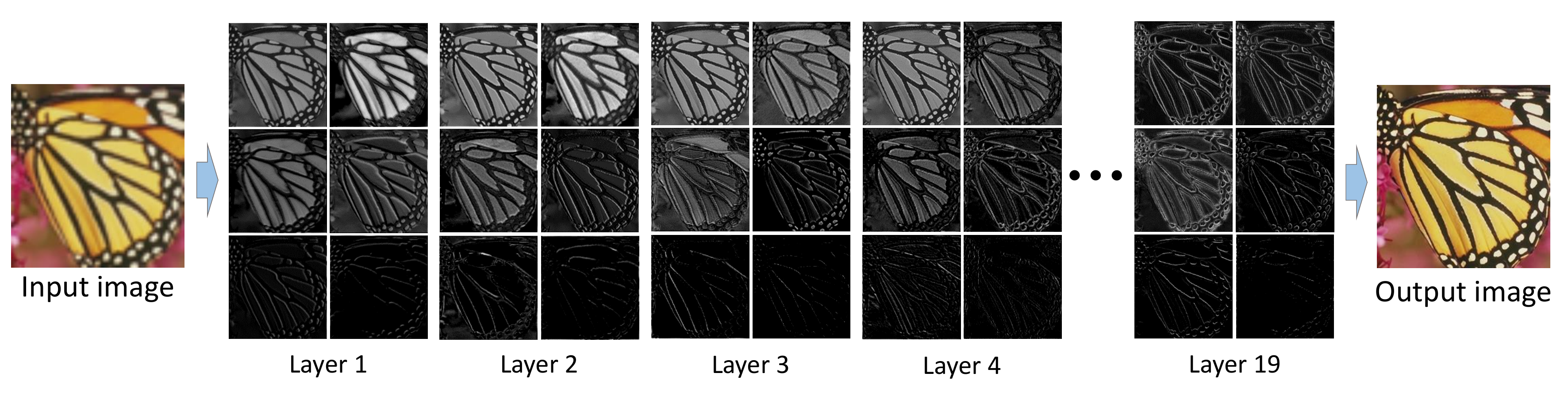}
}
\end{center}
\vspace{-0.55cm}
\caption{The output feature maps of different layers in VDSR~\cite{kim2016accurate} for the input image ``Butterfly''. The differences between any two adjacent layers are very similar with both texture and color information. The details \wrt the high-frequency information will be enhanced as the depth increases. These two important properties should be ensured by adder neural networks.}
\label{FigFeaturemaps}
\vspace{-0.3cm}
\end{figure*}

Here we firstly briefly introduce the single image super-resolution tasks using deep learning methods and then discuss the difficulty for directly using AdderNets to construct energy efficient SR models.

SISR aims at reconstructing a high-resolution image from a low-resolution image. It is a typical ill-posed reverse problem since an infinite number of high-resolution images could generate the same low-resolution image with downsampling. The objective function for the conventional SISR task can be formulated as: 

\vspace{-0.1cm}
\begin{equation}
\label{SrEquation}
\arg\min_{\mathbf{I_x}} \mathbf{I_y} = \mathbf{I_x}+\lambda\mathcal{R}(\mathbf{I_x})
\end{equation}
where $\mathbf{I_y}$ is the observation data, \ie, the low-resolution image, and $\mathbf{I_x}$ is the desired high-resolution image, $\mathcal{R}(\cdot)$ denotes the used priori such as smooth and additive noise, $\lambda$ is the tradeoff parameter.

Recently, Dong \etal~\cite{dong2014learning} first introduced deep learning method to super-resolution and achieved much better performance than traditional methods. Along with the improvement of super-resolution performance, the parameters and computations of SR networks grew rapidly, which seriously limited the efficiency of model executing on mobile devices. Hence, another research direction is to deploy efficient SR networks. Quantization, knowledge distillation, efficient operator designing and NAS have been explored to exploit efficient and accurate SR model. However, energy consumptions required by these portable SISR models are still much expensive for real-world mobile devices.

Distinguish to the existing model compression techniques such as pruning and quantization, Chen \etal~\cite{chen2019addernet} presented a  fundamentally different way to reduce the energy consumption by replacing multiplication by cheap addition. AdderNet with the same architecture as that of ResNet-50 obtains a 91.7\% top-5 accuracy while eliminating massive multiplications.  For the parameters \wrt the filters $\mathbf{W} \in \mathbb{R}^{k\times k\times c_{in}\times c_{out}}$ in an arbitrary layer in AdderNets, where $k$ is the size of kernel, $c_{in}$ and $c_{out}$ are the number of input channels and output channels, respectively. Denote the input feature as $\mathbf{X}\in \mathbb{R}^{h\times w \times c_{in}}$, where $w$ and $h$ are the width and height of the input feature, respectively. The original adder operation is:

\vspace{-0.2cm}
\begin{equation}
\label{AdderEquation}
\small
\mathbf{Y} (m,n,q) = -\sum_{i=1}^{k}\sum_{j=1}^{k}\sum_{g=1}^{c_{in}} \vert \mathbf{X} (m+i,n+j,g) - \mathbf{W} (i,j,g,q)\vert,
\end{equation}
where $|\cdot|$ is the absolute value function. $m$ and $n$ denote the spatial location of features. $q$ is the index of output channels. Although Eq.~\ref{AdderEquation} shows comparable performance on image classification tasks, the SISR problem defined in Eq.~\ref{SrEquation} is quite different to the conventional recognition task. For example, we need to ensure that the output results maintain the original texture in $\mathbf{X}(\ie,\mathbf{I_y})$, which cannot be easily learned by Eq.~\ref{AdderEquation}. Therefore, we should design a new framework of AdderNet for energy-efficient SISR models.
 
\subsection{Learning Identity Mapping using AdderNet}
\label{identitySection}

Usually, an arbitrary super-resolution model using neural network learns the mapping from input LR image to HR image in an end-to-end manner. In addition to enhancing the high-frequency details, the overall texture and color information should be also maintained. Figure~\ref{FigFeaturemaps} illustrates the feature maps of different convolutional layers for a given LR image using VDSR~\cite{kim2016accurate}. It can be seen that the difference between the input feature map and output feature map of each convolutional layer are very similar. This observation reveals that the identity mapping (\ie, $\mathbf{I_y} = \mathcal{F}(\mathbf{I_y})$) in SISR task using deep learning method is very essential. 

For the traditional convolutional network, the identity mapping is very easy to learn when the weight is an identity matrix. However, adder neural network defined in Eq.~\ref{AdderEquation} employs  $\ell_1$-norm to measure the distance between filter and input data. Although the  $\ell_1$ distance can perform well on the image classification task, the adder filter cannot approximate the identity mapping as described in Theorem~\ref{adder_theorem}. Stacking more adder layers may alleviate this problem, but it will increase the model size and complexity significantly.

\vspace{0.01cm}
\newtheorem{theorem}{Theorem}
\begin{theorem}
\label{adder_theorem}
	For the arbitrary low-resolution image $\mathbf I_y$, and an adder filter with its weight parameters  $\mathbf{W}$. There exists no $\mathbf{W}$ satisfying the following equation:
	\begin{equation}
	\mathbf{I_y} = \mathbf{I_y} \oplus \mathbf{W},
	\end{equation}
     where $\oplus$ denotes the adder operation defined in Eq.~\ref{AdderEquation}.
\end{theorem}
\begin{proof}
Here we consider a fully-connected layer for simplicity. Assuming that there exists a filter $\mathbf{W}$ that satisfies $\mathbf{I_y} = \mathbf{I_y} \oplus \mathbf{W}$ for any $\mathbf{I_y}\in \mathbb{R}^{d\times d\times c}$, where $\mathbf{W}\in\mathbb{R}^{d\times d \times c}$. Then for each element in $\mathbf{I_y}$, we have:
\begin{equation}
\small
\mathbf{I_y} (m,n,q) = -\sum_{i=1}^{d}\sum_{j=1}^{d}\sum_{g=1}^{c} \vert \mathbf{I_y} (m+i,n+j,g) - \mathbf{W} (i,j,g,q)\vert,
\end{equation}
where $m\in\{1,\dots,d\},n\in\{1,\dots,d\},q\in\{1,\dots,c\}$.

Then, we can select proper values for each element in $\mathbf{I^1_y}$ such that $\mathbf{I^1_y} (m,n,q) > \max_{i,j,k,g} {\vert \mathbf{W} (i,j,k,g)\vert}$, we have:
\begin{equation}
\small
\mathbf{I^1_y} (m,n,q) = -\sum_{i=1}^{d}\sum_{j=1}^{d}\sum_{g=1}^{c} [\mathbf{I^1_y} (m+i,n+j,g) - \mathbf{W} (i,j,g,q)].
\label{proof:1}
\end{equation}

Let $\mathbf{I^2_y} = \mathbf{I^1_y}+1 $, then we have:
\begin{equation}
\begin{aligned}
\small
\mathbf{I^1_y} (m,n,q) +1 = & -\sum_{i=1}^{d}\sum_{j=1}^{d}\sum_{g=1}^{c} [\mathbf{I^1_y} (m+i,n+j,g) + 1 \\& - \mathbf{W} (i,j,g,q)]\\
= & -\sum_{i=1}^{d}\sum_{j=1}^{d}\sum_{g=1}^{c} [\mathbf{I^1_y} (m+i,n+j,g) \\& -  \mathbf{W} (i,j,g,q)] - d^2 \times c. 
\end{aligned}
\label{proof:2}
\end{equation}

Combining Eq.~\ref{proof:1} and~\ref{proof:2}, we have $d^2 \times c + 1 = 0 $, which is obviously impossible. In addition, the above proof can be easily extended to convolution layers in which the filter size of $\mathbf{W}$ is often much smaller than that of the input data $\mathbf{I_y}$.

\end{proof}

According to Theorem~\ref{adder_theorem} and above analysis, the identity mapping cannot be directly learned using a one-layer adder neural network. To address this problem, we propose to refine the existing adder unit for adjusting the super-resolution task. In practice, we present a self-shortcut operation for each adder layer, \ie,

\begin{equation}
\label{residualmapping}
\mathbf{Y}^l = \mathbf{X}^l + \mathbf{W}^l \oplus \mathbf{X}^l.
\end{equation}
where $\mathbf{W}^l$ is the weights of adder filters in the $l$-th layer, $\mathbf{X}^l$ and  $\mathbf{Y}^l$ are the input data and output data, respectively. Since the output of Eq.~\ref{residualmapping} contains the input data $\mathbf{X}^l$ itself, we can utilize it to approximate the identity mapping by reducing the magnitude of $\mathbf{W}^l\oplus \mathbf{X}^l$.

Different to the image recognition tasks, features in most of SISR problems should maintain a fixed size, \ie, the width and height of $\mathbf{X}^l$ and $\mathbf{Y}^l$ are exactly the same. Thus, the new calculation as described in Eq.~\ref{residualmapping} can be embedded into most of conventional SISR models. In the following experiments, we will utilize Eq.~\ref{residualmapping} to replace most of convolutional layers whose output size is the same as that of the input size.

\subsection{Learnable Power Activation}
\label{PowerActivation}

Besides the identity mapping, there is another important functionality of traditional convolution filters that cannot be easily ensured by adder filters. The goal of a SISR model is to enhance the details including color and texture information of the input low-resolution images. Therefore, the high-pass filter is also a very important component in most of existing SISR models. It can be found in Figure~\ref{FigFeaturemaps}, the details of the input image are gradually enhanced when the network depth is increased.

Generally, natural images are composed of different-frequency information. For example, the background and a large area of grass are low-frequency information, where most of neighbor pixels are very closed. In contrast, edges of objects and some buildings are exactly high-frequency information for the given entire image. In practice, if we define an arbitrary image as the combination of high-frequency part and the low-frequency part as $\mathbf{I} =\mathbf{I}_H+\mathbf{I}_L$, an ideal high-pass filter $\Phi(\cdot)$ used in the super-resolution task and other image processing problems could be defined as 

\begin{equation}
\label{HighFrequencyEquation}
\small
\Phi(\mathbf{I}) = \Phi(\mathbf{I}_H+\mathbf{I}_L) = \sigma(\mathbf{I}_H),
\end{equation}
which only preserves the high-frequency part of the input image. Wherein, $\sigma(\mathbf{I}_H)$ is the convolution response of the high-frequency part. The above equation can help the SISR model removing redundant outputs and noise and enhancing the high-frequency details, which is also a very essential component in the SISR models.

Similarly, for the traditional convolution operation, the functionality of Eq.~\ref{HighFrequencyEquation} can be directly implemented. For example, a $2 \times 2$ high-pass filter $\small \begin{bmatrix} -1 & +1 \\ +1 & -1 \end{bmatrix}$ can be used for removing any flat areas in $\mathbf{I}$. However, for the adder neural network as defined in Eq.~\ref{AdderEquation}, it is impossible to achieve the functionality as described in Eq.~\ref{HighFrequencyEquation}.

\begin{theorem}
\label{filter_theorem}
	Let $\mathbf{E}\in\mathbb{R}^{d \times d}$ be an input image for the given SR model using adder operation, where each element in $\mathbf{E}$ is equal to $1$. $\mathbf{W}$ denotes the weights of an arbitrary adder filter. There exists no $\mathbf{W}\in\mathbb{R}^{d \times d}$ and a constant $a\in \mathbb{R}$ satisfying the following equation:
	\begin{equation}
	(s* \mathbf{E}) \oplus \mathbf{W} = a,
	\end{equation}
	where $s\in\mathbb{R}$, and $\oplus$ denotes the adder operation which contains only additions as defined in ~\ref{AdderEquation}.
\end{theorem}
\begin{proof}
Assuming that there exists an adder filter $\mathbf{W}$ that satisfies $(s* \mathbf{E}) \oplus \mathbf{W} = a$, for any $s\in\mathbb{R}$. We can find a $s$ such that $s > max_{i,j} {\vert \mathbf{W} (i,j)\vert}$. Then we have:
\begin{equation}
\small
\begin{aligned}
-\sum_{i=1}^{d}\sum_{j=1}^{d} [ (s+1)\mathbf{E} (i,j) - \mathbf{W} (i,j)] = a = & \\ -\sum_{i=1}^{d}\sum_{j=1}^{d} [ s\mathbf{E} (i,j) - \mathbf{W} (i,j)] &,
\label{proof:3}
\end{aligned}
\end{equation}
which means $\sum_{i=1}^{d}\sum_{j=1}^{d}  (s+1)\mathbf{E} (i,j) = \sum_{i=1}^{d}\sum_{j=1}^{d} s\mathbf{E} (i,j)$ and leads to a contradiction.

\end{proof}

According to the above theorem, the functionality of high-pass filter cannot be replaced by adder filter. Thus, the super-resolution process in the adder neural networks will involve more redundancy. To this end, we need to develop a new scheme for using adder neural networks to conduct the SISR tasks that can compensate this defect.

Admittedly, we can also add some parameters and filters to improve the capacity of the SR model using adder units, the save on energy and calculation will be reduced. Fortunately, Sharabati and Xi~\cite{sharabati2016fast} applied the Box-Cox transformation~\cite{sakia1992box} in the image denoising task and found that this transformation can achieve the similar functionality to that of the high-pass filters without adding massive parameters and calculations. Oliveira~\etal~\cite{oliveira2018single} further discussed the functionality of Box-Cox transformation in image super-resolution. In addition, a sign-preserving power law point transformation is also explored for emphasizing the areas with abundant details in the input image~\cite{ramponi1996nonlinear}. Therefore, we propose a learnable power activation function to solve the defect of AdderNet and refine the output images, \ie, 
\begin{equation}
\label{powerfunction}
\mathcal{P}(\mathbf{Y}) = \sign (\mathbf{Y}) \cdot {|\mathbf{Y}|}^\alpha ,
\end{equation}
where $\mathbf{Y}$ is the output features, $\sign(\cdot)$ is the sign function, $\alpha > 0$ is a learnable parameter for adjusting the information and distribution. When $\alpha > 1$, the above activation function can enhance the contrast of output images and emphasize the high-frequency information. When $0 < \alpha < 1$, Eq.~\ref{powerfunction} can smooth all signals in the output image and remove artifacts and noise. In addition, the above function can be easily embedded into the conventional ReLU in any SISR models.

By exploiting these two methods described in Eq.~\ref{residualmapping}  and Eq.~\ref{powerfunction}, we can address the aforementioned problems on using adder networks for conducting the SISR task. Although there are some additional computations introduced, \eg, Eq.~\ref{residualmapping} needs a shortcut to maintain the information in the input data, and the learnable parameter $\alpha$ in Eq.~\ref{powerfunction} leads to some additional multiplications. However, compared with the massive operations required by either convolutional layer or adder layer, they are very trivial. For example, the number of adder operations defined in Eq.~\ref{AdderEquation} is about $2k^2c_{in}c_{out}hw$, and the additional computations required by  Eq.~\ref{residualmapping} and Eq.~\ref{powerfunction} are both equal to $c_{out}hw$. Considering that $2k^2c_{in}$ is usually a relatively large value in modern deep neural architectures (\eg, $k = 3$ and $c_{in} = 64$), the additional computation for each layer is over $1000\times$ lower than that of the original method. In the next section, we will conduct extensive experiments to illustrate the superiority of the proposed method on both the visual quality and energy consumptions.

\section{Experiments}
\label{experiment}

In the above section, we have developed a series of new operations for establishing SISR models using adder neural networks. Here we will conduct experiments to verify the effectiveness of the proposed AdderSR networks.

\vspace{-0.25cm}
\paragraph{Datasets.} 

To evaluate the performance of adder neural networks on super-resolution tasks, we select several benchmark image datasets to conduct the experiments. Following the setting of VDSR~\cite{kim2016accurate}, \emph{291} dataset is employed to train adder VDSR networks. It consists of 91 images from Yang \etal~\cite{yang2010image} and 200 images from Berkeley Segmentation Dataset~\cite{martin2001database}. In addition, DIV2K dataset~\cite{timofte2017ntire} is utilized to train adder EDSR networks in the following experiments. This dataset consists of 800 training images and 100 validation images. In order to compare with other state-of-the-art methods, four relatively small benchmarks are also selected including Set5, Set14, B100 and Urban100. The LR image is generated with bicubic downsampling. Super-resolution results are evaluated using both peak signal-to-noise ratio (PSNR~\cite{dong2014learning}) and structure similarity index (SSIM~\cite{wang2004image}) on Y channel ( \ie, luminance) of YCbCr space.

\vspace{-0.25cm}
\paragraph{Training setting.}

We evaluate the performance of the proposed AdderSR networks using two famous neural architectures for super-resolution, \ie,  VDSR~\cite{kim2016accurate} and EDSR~\cite{lim2017enhanced}, which are shown extraordinary performance for generating images with high visual quality. Following the setting of the conventional AdderNet, we do not replace the first and the last convolutional layers in these networks, and the batch normalization is employed on each adder layer. To accelerate convergence speed of AdderSR networks, the learning rate for adder layers in our models is enlarged 10 times than that of convolutional layers in baselines. Specifically, the learning rates of adder layer and convolutional layer are initialized as $3\times10^{-3}$ and $3\times 10^{-4}$, respectively. The optimizer utilized in adder neural network for super-resolution task is ADAM~\cite{kingma2014adam} due to its fast convergence speed. Hyper-parameters here are set as $\beta_{1}=0.9$, $\beta_{2}=0.999$ and $\epsilon=10^{-8}$. Other settings such as patch size and data augmentation strategies are totally the same as those in baseline VDSR and EDSR.

\begin{figure*}[tp]
\setlength{\abovecaptionskip}{-0.1cm}
\begin{center}
\scalebox{0.998}{
\includegraphics[width=1\textwidth]{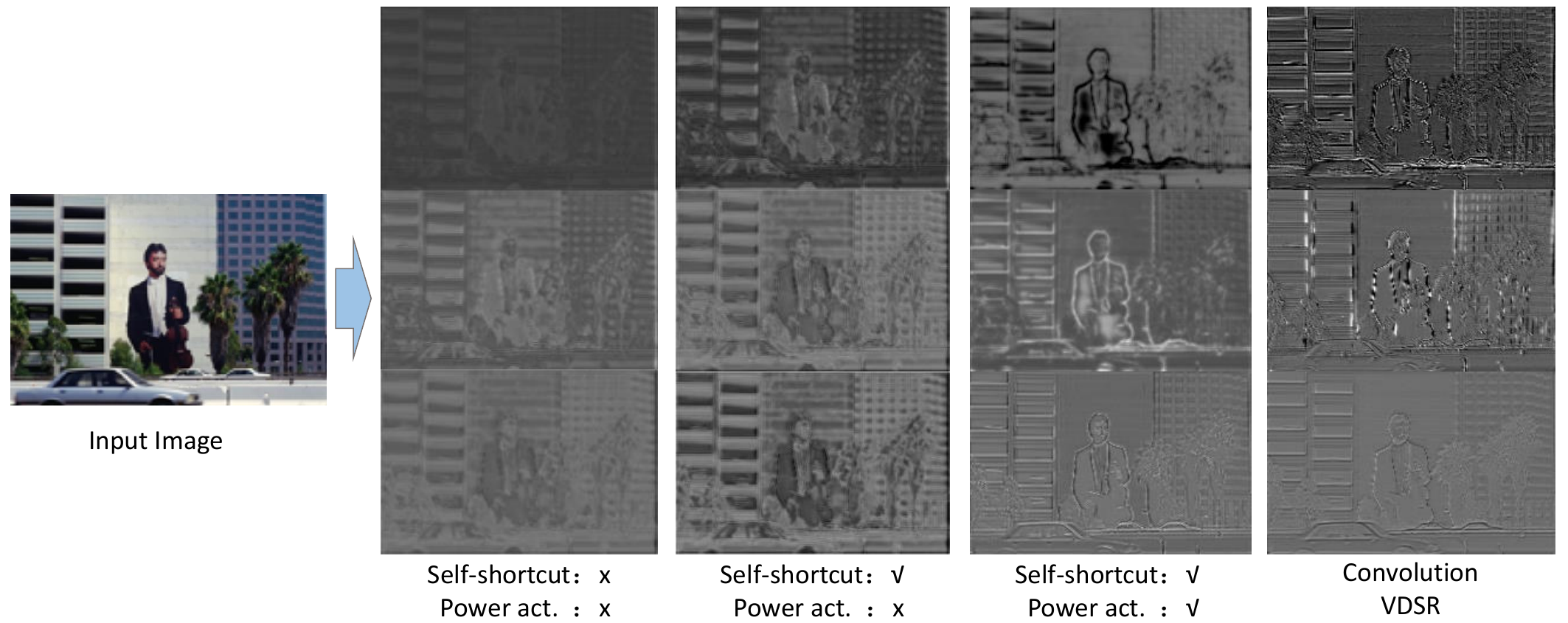}
}
\end{center}
\vspace{-0.45cm}
\caption{The output feature maps of adder layers in AdderSR networks with different strategies. Feature maps of adder layer without self-shortcut (Eq.~\ref{residualmapping}) cannot maintain the overall texture information of their input features. The power activation function (Eq.~\ref{powerfunction}) effectively enhances high-frequency regions.}
\label{IdentityFeature}
\vspace{-0.10cm}
\end{figure*}

\begin{table*}[t]
\small
\begin{center}
\caption{Ablation study of the proposed AdderSR network. Wherein, results of VDSR model on the four datasets are directly reported from the original paper ~\cite{kim2016accurate}.}
\label{ablationTable}
\renewcommand{\arraystretch}{1.3}
\setlength{\tabcolsep}{1.0mm}{
\begin{tabular}{p{2.4cm}<{\centering}|p{2.6cm}<{\centering}|p{2.6cm}<{\centering}|p{1.4cm}<{\centering}| p{1.4cm}<{\centering} | p{1.4cm}<{\centering} | p{1.4cm}<{\centering}}
\hline
\multirow{2}*{Model} & \multicolumn{2}{c|}{Architecture} & Set5 & Set14 & B100 & Urban100 \\
\cline{2-3}
~ &  Self-shortcut (Eq.~\ref{residualmapping}) & Power Act. (Eq.~\ref{powerfunction}) & PSNR &  PSNR & PSNR & PSNR \\
\hline
Adder VDSR 1 & $\times$ & $\times$ & 35.11 & 31.40 & 30.63 & 27.77\\
\hline
Adder VDSR 2 & $\times$ & \checkmark & 35.44 & 31.62 & 30.78 & 27.96\\
\hline
Adder VDSR 3 & \checkmark & $\times$ & 37.10 & 32.77 & 31.69 & 30.05\\
\hline
Adder VDSR 4 & \checkmark & \checkmark & 37.37 & 32.93 & 31.81 & 30.48\\
\hline
Conv. VDSR~\cite{kim2016accurate} & $-$ & $ - $ & 37.53 & 33.03 & 31.90 & 30.76\\
\hline
\end{tabular}}
\end{center}
\vspace{-6mm}
\end{table*}

\paragraph{Ablation study.}

To ensure the performance of adder neural networks on the SISR tasks, we have thoroughly designed two new operations in Eq.~\ref{residualmapping} and Eq.~\ref{powerfunction}. Here we will first conduct the detailed ablation study to illustrate their functionalities.
In practice, the VDSR using convolutional layers is selected as baseline, and we replace all intermediate layers by adder layers as described in Eq.~\ref{AdderEquation}. Table~\ref{ablationTable} shows results of ablation experiments. Without Eq.~\ref{residualmapping} and Eq.~\ref{powerfunction}, the PSNR of AdderSR is $2.08$ \emph{dB} lower than that of conventional VDSR network on the average of four benchmark datasets. Such a PSNR value decline will increase the artifacts in the resulting high-resolution images. In Eq.~\ref{AdderEquation}, the self-shortcut makes it possible to optimize an identity mapping. Thus, the performance of AdderSR network using Eq.~\ref{AdderEquation} can obtain an about 1.67 \emph{dB} PSNR enhancement on average. 
In addition, Eq.~\ref{powerfunction} is proposed to emphasize the high-frequency information in intermediated features and reconstructed images. By embedding Eq.~\ref{powerfunction} into AdderSR network, the PSNR value can be further improved with about 0.24 \emph{dB} on benchmark datasets. If both Eq.~\ref{residualmapping} and Eq.~\ref{powerfunction} are employed on AdderSR network, the performance can be improved with 1.91 \emph{dB} on the average of four datasets. The final result of AdderSR network is pretty close to that of conventional VDSR. Detailed visualization results of these models can be found in the supplementary materials.

\paragraph{Feature Visulizations.}

To have an explicit illustration on the functionalities of different components in our models using adder layers, we visualize the feature maps of Adder VDSR model 1, 3 and 4 reported in Table~\ref{ablationTable}. From Figure~\ref{IdentityFeature}, we can see that the feature maps of AdderSR network without Eq.~\ref{residualmapping} are quite different from the input image. The texture information such as human face, cars and trees in the input image are blured or distorted. In contrast, the features maps of AdderSR network using Eq.~\ref{residualmapping} preserve more details information than those of the model without Eq.~\ref{residualmapping}. Moreover, it is obvious that the high-frequency information (\eg, edge, corner) is emphasized while a portion of low-frequency information is eliminated with the help of Eq.~\ref{powerfunction}. That is a very important functionality for proceeding high-resolution images. In summary, by exploiting the proposed method, we can generate features with abundant texture and establish effective SISR models using only additions.

\begin{figure*}[tp]
\setlength{\abovecaptionskip}{-0.1cm}
\begin{center}
\scalebox{0.92}{
\includegraphics[width=1\textwidth]{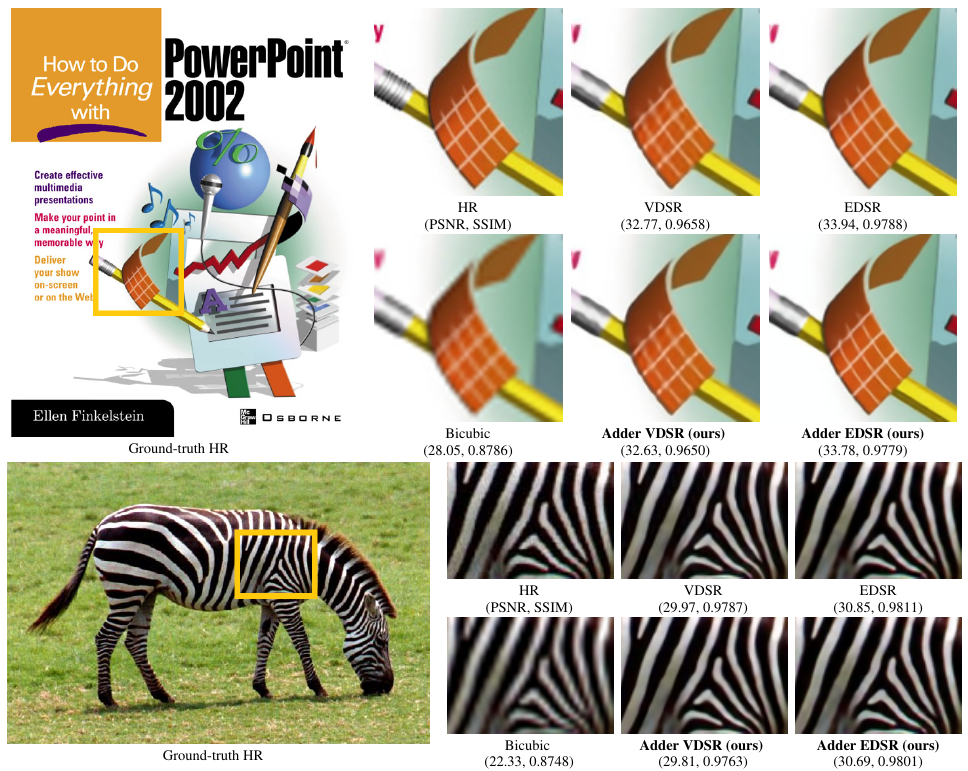}
}
\end{center}
\vspace{-0.35cm}
\caption{Visualization of super-resolution image of AdderSR Network and CNN on $\times 3$ scale.}
\label{SRImageCompare}
\vspace{-0.3cm}
\end{figure*}

\paragraph{Overall Comparison.}

After investigating the impact of different components of the proposed method, we then conduct the experiments using the benchmark VDSR and EDSR architectures with different scaling factors, \ie, $\times 2$, $\times 3$ and $\times 4$ to illustrate the superiority of the proposed AdderSR networks for the SISR problem. Table~\ref{compareSota} shows the average PSNR and SSIM results of different models using both convolutional layers and adder units on Set5, Set14, B100, Urban100, respectively. It can be found in Table~\ref{compareSota}, our models using only additions in intermediate layers can also achieve comparable PSNR and SSIM values to those of their baselines with massive multiplications. The average PSNR gap between AdderSR models and baselines using convolutions over the four datasets is about 0.17 \emph{dB}.

Besides the quantitative comparison, we also provide the comparison on the visual qualities of the same architecture using both additions and multiplications as shown in Figure~\ref{SRImageCompare}. Actually, since the PSNR values of AdderSR networks and their baselines as shown in Table~\ref{compareSota} are very closed, the output high-resolution images are of the similar visual quality. Particularly, the network using adder units can also effectively reconstruct the texture and color information. These results demonstrate that, besides the visual recognition task, we can also utilize additions to establish excellent neural networks for solving the SISR image processing problem. More visualization results of these models can be found in the supplementary materials.

\begin{table*}[t]
\small
\begin{center}
\caption{Quantitative results of baseline convolutional networks and our AdderSR models. Wherein,  $\times 2$, $\times 3$ and $\times 4$ are the output scaling factors for the SISR task. ANN and CNN denote the networks using adder units and traditional convolution layers, respectively. "G" represents $10^9$.}
\label{compareSota}
\renewcommand{\arraystretch}{1.2}
\setlength{\tabcolsep}{1.6mm}{
\begin{tabular}{@{}  p{0.85cm}<{\centering}|p{1.05cm}<{\centering} p{0.8cm}<{\centering}|p{1.1cm}<{\centering} p{1.1cm}<{\centering}|p{1.75cm}<{\centering}p{1.75cm}<{\centering} p{1.75cm}<{\centering} p{1.75cm}<{\centering}}
\hline
\multirow{2}*{Scale} & \multirow{2}*{Model}  & \multirow{2}*{Type} & \#Mul. & \#Add. &  Set5 & Set14 & B100 & Urban100 \\
~  & ~ & ~& (G) &  (G) &  PSNR/SSIM & PSNR/SSIM & PSNR/SSIM & PSNR/SSIM \\
\hline
\multirow{4}*{$\times 2$} & \multirow{2}*{VDSR} & ANN & 1.1 & 1224.1 & 37.37/0.9583 & 32.93/0.9115 & 31.81/0.8945 & 30.48/0.9104\\
~ & ~ & CNN & 612.6 & 612.6 & 37.53/0.9587 & 33.03/0.9124 & 31.90/0.8960 & 30.76/0.9140\\
\cline{2-9}
~ & \multirow{2}*{EDSR} & ANN & 7.9 & 18489.8 & 37.92/0.9589 & 33.82/0.9183 & 32.23/0.9000 & 32.63/0.9309\\
~ & ~ & CNN & 9248.9 & 9248.9 & 38.11/0.9601 & 33.92/0.9195 & 32.32/0.9013 & 32.93/0.9351\\
\hline
\multirow{4}*{$\times 3$} & \multirow{2}*{VDSR} & ANN & 1.1 & 1224.1 & 33.54/0.9204 & 29.75/0.8306 & 28.76/0.7958 & 26.95/0.8220\\
~ & ~ & CNN & 612.6 & 612.6 & 33.66/0.9213 & 29.77/0.8314 & 28.82/0.7976 & 27.14/0.8279\\
\cline{2-9}
~ & \multirow{2}*{EDSR} & ANN& 7.1 & 8825.2 & 34.35/0.9212 & 30.33/0.8420 & 29.13/0.8068 & 28.54/0.8555\\
~ & ~ & CNN & 4416.1 & 4416.1 & 34.65/0.9282 & 30.52/0.8462 & 29.25/0.8093 & 28.80/0.8653\\
\hline
\multirow{4}*{$\times 4$} & \multirow{2}*{VDSR} & ANN & 1.1 & 1224.1 & 31.19/0.8804 & 27.93/0.7646 & 27.22/0.7225 & 25.02/0.7447\\
~ & ~ & CNN & 612.6 & 612.6 & 31.35/0.8838 & 28.01/0.7674 & 27.29/0.7251 & 25.18/0.7524\\
\cline{2-9}
~ & \multirow{2}*{EDSR} & ANN & 6.8 & 5442.6 & 32.13/0.8864 & 28.57/0.7800 & 27.58/0.7368 & 26.33/0.7874\\
~ & ~ & CNN & 2724.7 & 2724.7 & 32.46/0.8968 & 28.80/0.7876 & 27.71/0.7420 & 26.64/0.8033\\
\hline
\end{tabular}}
\vspace{-0.4cm}
\end{center}
\end{table*}

\begin{table*}[t]
\small
\begin{center}
\caption{Quantitative results of CARN and pruned architectures using convolutional and adder layers, respectively. PSNR values are reported on Urban 100 database.}
\label{efficientTable}
\renewcommand{\arraystretch}{1.3}
\setlength{\tabcolsep}{1.0mm}{
\begin{tabular}{p{2.0cm}<{\centering}||p{1.6cm}<{\centering}|p{1.6cm}<{\centering}|p{1.6cm}<{\centering}| p{1.6cm}<{\centering} || p{1.6cm}<{\centering}| p{1.6cm}<{\centering}| p{1.6cm}<{\centering}| p{1.6cm}<{\centering}}
\hline
Model & CARN & CARN-$\frac{2}{3}$ & CARN-$\frac{1}{2}$ & CARN-$\frac{1}{4}$ & A-CARN & A-CARN-$\frac{2}{3}$ & A-CARN-$\frac{1}{2}$ & A-CARN-$\frac{1}{4}$ \\
\hline
\hline
PSNR(\emph{dB}) & 31.92 & 31.35 & 31.15 & 30.40 & 31.74 & 31.27 & 30.94 & 30.21\\
\hline
Energy(\emph{pJ}) & 1027G & 404G & 260G & 66G & 402G & 160G & 102G & 26G\\
\hline
\end{tabular}}
\end{center}
\vspace{-6.mm}
\end{table*}

\paragraph{Energy Consumptions.}

\begin{table}[ht]
\small
\begin{center}
\caption{The energy consumptions of different networks for $\times 2$ scale. The energy cost is computed using 720p (\ie $1280 \times 720$) high-resolution image.}
\label{energyTable}
\renewcommand{\arraystretch}{1.15}
\setlength{\tabcolsep}{1.40mm}{
\begin{tabular}{@{}  b{2.1cm}<{\centering}|b{1.2cm}<{\centering}| b{1.2cm}<{\centering}|b{1.2cm}<{\centering}| b{1.2cm}<{\centering}}
\hline
\multirow{2}*{Model} & \multicolumn{2}{c|}{VDSR}  & \multicolumn{2}{c}{EDSR} \\
\cline{2-5}
~  & CNN & ANN & CNN &  ANN  \\
\hline
Energy cost (\emph{pJ}) & 2817.9G & 1105.6G & 42544.7G & 16670.3G \\
\hline
\end{tabular}}
\end{center}
\vspace{-1.0cm}
\end{table}

As discussed in previous works~\cite{chen2019addernet,horowitz20141,sze2017efficient,you2020shiftaddnet},  additions require much lower energy consumptions compared with multiplications. We further calculate the energy consumptions of different networks used in Table~\ref{compareSota}. In practice, most of values in recent SISR models are 32-bit floating numbers, and the energy consumptions for a 32-bit addition and multiplication are $0.9$ \emph{pJ} and $3.7$ \emph{pJ}, respectively. The amount of remaining multiplication operations in AdderSR network is extremely small compared with the FLOPs of the entire network. The detailed energy costs of the these networks are reported in Table~\ref{energyTable}, which is computed according to literature~\cite{horowitz20141,sze2017efficient}. Obviously, the proposed AdderSR method can reduce the energy cost for reconstructing a $1280 \times 720$ image by a fact of about 2.5$\times$. If we further quantize the weights and activations of these models shown in Table~\ref{compareSota} to int8 values, we can obtain an about 3.8$\times$ reduction on the energy consumption using the proposed AdderSR with comparable performance. These results can be found in supplementary materials. We also conduct experiments using our methods on recent lightweight SR models (\ie, CARN~\cite{ahn2018fast}) and report the energy consumption v.s. performance comparison in Table~\ref{efficientTable}, where CARN-$\frac{2}{3}$ , CARN-$\frac{1}{2}$ and CARN-$\frac{1}{4}$ are the models with $\frac{2}{3}$, $\frac{1}{2}$ and $\frac{1}{4}$ channels after applying filter pruning, respectively. A-CARN denotes the model using AdderNets. We can achieve the comparable PSNR values on compact SISR architectures. The performance of A-CARN is about 0.4dB higher than that of CARN-$\frac{2}{3}$ with similar energy consumption (404G\emph{pJ}) after pruning. This significant reduction on energy consumption will make these deep learning models portable on mobile devices.
\section{Conclusions and Discussion}
\label{conclusion}

This paper investigates the single image super-resolution problem using AdderNets. Without changing the original neural architectures, we develop a new adder unit and a novel learnable power activation for addressing the defects in existing adder neural networks. The new AdderSR models can learn the functionalities of conventional identity mapping and high-pass filter, which are essential for providing images with high visual quality. Experimental results conducted on several benchmarks illustrate that, the proposed AdderSR networks can achieve the similar visual quality to that of their baselines using traditional convolution filters. Meanwhile, since additions are much cheaper than multiplications, we can reduce the energy consumptions of these models by about 2.5$\times$. Besides the image super-resolution task, the techniques in this paper can be well transferred to other image processing problems including denoising, deblurring, \etc Future works will focus on more applications and low-bit quantization versions of these networks to achieve higher reduction on the energy consumption.

{\small
\bibliographystyle{ieee_fullname}
\bibliography{ref}

\begin{thebibliography}{10}\itemsep=-1pt

\bibitem{ahn2018fast}
Namhyuk Ahn, Byungkon Kang, and Kyung-Ah Sohn.
\newblock Fast, accurate, and lightweight super-resolution with cascading
  residual network.
\newblock In {\em ECCV}, pages 252--268, 2018.

\bibitem{chen2019addernet}
Hanting Chen, Yunhe Wang, Chunjing Xu, Boxin Shi, Chao Xu, Qi Tian, and Chang
  Xu.
\newblock Addernet: Do we really need multiplications in deep learning?
\newblock In {\em CVPR}, 2020.

\bibitem{courbariaux2015binaryconnect}
Matthieu Courbariaux, Yoshua Bengio, and Jean-Pierre David.
\newblock Binaryconnect: Training deep neural networks with binary weights
  during propagations.
\newblock In {\em NeurIPS}, pages 3123--3131, 2015.

\bibitem{dong2014learning}
Chao Dong, Chen~Change Loy, Kaiming He, and Xiaoou Tang.
\newblock Learning a deep convolutional network for image super-resolution.
\newblock In {\em ECCV}, pages 184--199, 2014.

\bibitem{dong2016accelerating}
Chao Dong, Chen~Change Loy, and Xiaoou Tang.
\newblock Accelerating the super-resolution convolutional neural network.
\newblock In {\em ECCV}, pages 391--407. Springer, 2016.

\bibitem{fu2020autogan}
Yonggan Fu, Wuyang Chen, Haotao Wang, Haoran Li, Yingyan Lin, and Zhangyang
  Wang.
\newblock Autogan-distiller: Searching to compress generative adversarial
  networks.
\newblock In {\em ICML}, 2020.

\bibitem{gao2018image}
Qinquan Gao, Yan Zhao, Gen Li, and Tong Tong.
\newblock Image super-resolution using knowledge distillation.
\newblock In {\em ACCV}, pages 527--541. Springer, 2018.

\bibitem{gong2019autogan}
Xinyu Gong, Shiyu Chang, Yifan Jiang, and Zhangyang Wang.
\newblock Autogan: Neural architecture search for generative adversarial
  networks.
\newblock In {\em ICCV}, pages 3224--3234, 2019.

\bibitem{guo2016dynamic}
Yiwen Guo, Anbang Yao, and Yurong Chen.
\newblock Dynamic network surgery for efficient dnns.
\newblock In {\em NeurIPS}, pages 1379--1387, 2016.

\bibitem{han2020ghostnet}
Kai Han, Yunhe Wang, Qi Tian, Jianyuan Guo, Chunjing Xu, and Chang Xu.
\newblock Ghostnet: More features from cheap operations.
\newblock In {\em CVPR}, pages 1580--1589, 2020.

\bibitem{he2016deep}
Kaiming He, Xiangyu Zhang, Shaoqing Ren, and Jian Sun.
\newblock Deep residual learning for image recognition.
\newblock In {\em CVPR}, pages 770--778, 2016.

\bibitem{hinton2015distilling}
Geoffrey Hinton, Oriol Vinyals, and Jeff Dean.
\newblock Distilling the knowledge in a neural network.
\newblock {\em NeurIPS workshops}, pages 1--9, 2015.

\bibitem{horowitz20141}
Mark Horowitz.
\newblock Computing's energy problem (and what we can do about it).
\newblock In {\em 2014 IEEE International Solid-State Circuits Conference
  Digest of Technical Papers (ISSCC)}, pages 10--14, 2014.

\bibitem{hou2020efficient}
Zejiang Hou and Sun-Yuan Kung.
\newblock Efficient image super resolution via channel discriminative deep
  neural network pruning.
\newblock In {\em ICASSP}, pages 3647--3651. IEEE, 2020.

\bibitem{hui2018fast}
Zheng Hui, Xiumei Wang, and Xinbo Gao.
\newblock Fast and accurate single image super-resolution via information
  distillation network.
\newblock In {\em CVPR}, pages 723--731, 2018.

\bibitem{kim2016accurate}
Jiwon Kim, Jung Kwon~Lee, and Kyoung Mu~Lee.
\newblock Accurate image super-resolution using very deep convolutional
  networks.
\newblock In {\em CVPR}, pages 1646--1654, 2016.

\bibitem{kim2016deeply}
Jiwon Kim, Jung Kwon~Lee, and Kyoung Mu~Lee.
\newblock Deeply-recursive convolutional network for image super-resolution.
\newblock In {\em CVPR}, pages 1637--1645, 2016.

\bibitem{kingma2014adam}
Diederik~P Kingma and Jimmy Ba.
\newblock Adam: A method for stochastic optimization.
\newblock In {\em ICLR}, 2014.

\bibitem{lai2017deep}
Wei-Sheng Lai, Jia-Bin Huang, Narendra Ahuja, and Ming-Hsuan Yang.
\newblock Deep laplacian pyramid networks for fast and accurate
  super-resolution.
\newblock In {\em CVPR}, pages 624--632, 2017.

\bibitem{Li2020PAMS}
Huixia Li, Chenqian Yan, Shaohui Lin, Xiawu Zheng, Yuchao Li, Baochang Zhang,
  Fan Yang, and Rongrong Ji.
\newblock Pams: Quantized super-resolution via parameterized max scale.
\newblock In {\em ECCV}, 2020.

\bibitem{lim2017enhanced}
Bee Lim, Sanghyun Son, Heewon Kim, Seungjun Nah, and Kyoung Mu~Lee.
\newblock Enhanced deep residual networks for single image super-resolution.
\newblock In {\em CVPRW}, pages 136--144, 2017.

\bibitem{lin2020hrank}
Mingbao Lin, Rongrong Ji, Yan Wang, Yichen Zhang, Baochang Zhang, Yonghong
  Tian, and Ling Shao.
\newblock Hrank: Filter pruning using high-rank feature map.
\newblock In {\em CVPR}, pages 1529--1538, 2020.

\bibitem{ma2019efficient}
Yinglan Ma, Hongyu Xiong, Zhe Hu, and Lizhuang Ma.
\newblock Efficient super resolution using binarized neural network.
\newblock In {\em CVPRW}, pages 0--0, 2019.

\bibitem{martin2001database}
David Martin, Charless Fowlkes, Doron Tal, and Jitendra Malik.
\newblock A database of human segmented natural images and its application to
  evaluating segmentation algorithms and measuring ecological statistics.
\newblock In {\em ICCV}, volume~2, pages 416--423, 2001.

\bibitem{muqeet2020ultra}
Abdul Muqeet, Jiwon Hwang, Subin Yang, Jung~Heum Kang, Yongwoo Kim, and Sung-Ho
  Bae.
\newblock Ultra lightweight image super-resolution with multi-attention layers.
\newblock In {\em ECCV Workshop}, 2020.

\bibitem{oliveira2018single}
N{\'\i}cholas Andr{\'e} Pinho~de Oliveira et~al.
\newblock Single image super-resolution method based on linear regression and
  box-cox transformation.
\newblock 2018.

\bibitem{ramponi1996nonlinear}
Giovanni Ramponi, Norbert~K Strobel, Sanjit~K Mitra, and Tian-Hu Yu.
\newblock Nonlinear unsharp masking methods for image contrast enhancement.
\newblock {\em Journal of electronic imaging}, 5(3):353--367, 1996.

\bibitem{sakia1992box}
Remi~M Sakia.
\newblock The box-cox transformation technique: a review.
\newblock {\em Journal of the Royal Statistical Society: Series D (The
  Statistician)}, 41(2):169--178, 1992.

\bibitem{sharabati2016fast}
Walid~K Sharabati and Bowei Xi.
\newblock Fast local polynomial regression approach for speckle noise removal.
\newblock In {\em ICPR}, pages 3198--3203, 2016.

\bibitem{song2019efficient}
Dehua Song, Chang Xu, Xu Jia, Chunjing Xu, and Yunhe Wang.
\newblock Efficient residual dense block search for image super-resolution.
\newblock In {\em AAAI}, 2020.

\bibitem{sze2017efficient}
Vivienne Sze, Yu-Hsin Chen, Tien-Ju Yang, and Joel~S Emer.
\newblock Efficient processing of deep neural networks: A tutorial and survey.
\newblock {\em Proceedings of the IEEE}, 105(12):2295--2329, 2017.

\bibitem{tai2017image}
Ying Tai, Jian Yang, and Xiaoming Liu.
\newblock Image super-resolution via deep recursive residual network.
\newblock In {\em CVPR}, pages 3147--3155, 2017.

\bibitem{tai2017memnet}
Ying Tai, Jian Yang, Xiaoming Liu, and Chunyan Xu.
\newblock Memnet: A persistent memory network for image restoration.
\newblock In {\em ICCV}, pages 4539--4547, 2017.

\bibitem{timofte2017ntire}
Radu Timofte, Eirikur Agustsson, Luc Van~Gool, Ming-Hsuan Yang, and Lei Zhang.
\newblock Ntire 2017 challenge on single image super-resolution: Methods and
  results.
\newblock In {\em CVPRW}, pages 114--125, 2017.

\bibitem{wang2020gan}
Haotao Wang, Shupeng Gui, Haichuan Yang, Ji Liu, and Zhangyang Wang.
\newblock Gan slimming: All-in-one gan compression by a unified optimization
  framework.
\newblock In {\em ECCV}, 2020.

\bibitem{wang2004image}
Zhou Wang, Alan~C Bovik, Hamid~R Sheikh, Eero~P Simoncelli, et~al.
\newblock Image quality assessment: from error visibility to structural
  similarity.
\newblock {\em TIP}, 13(4):600--612, 2004.

\bibitem{xin2020binarized}
Jingwei Xin, Nannan Wang, Xinrui Jiang, Jie Li, Heng Huang, and Xinbo Gao.
\newblock Binarized neural network for single image super resolution.
\newblock In {\em ECCV}, 2020.

\bibitem{xu2020kernel}
Yixing Xu, Chang Xu, Xinghao Chen, Wei Zhang, Chunjing Xu, and Yunhe Wang.
\newblock Kernel based progressive distillation for adder neural networks.
\newblock In {\em NeurIPS}, 2020.

\bibitem{yang2010image}
Jianchao Yang, John Wright, Thomas~S Huang, and Yi Ma.
\newblock Image super-resolution via sparse representation.
\newblock {\em TIP}, 19(11):2861--2873, 2010.

\bibitem{you2020shiftaddnet}
Haoran You, Xiaohan Chen, Yongan Zhang, Chaojian Li, Sicheng Li, Zihao Liu,
  Zhangyang Wang, and Yingyan Lin.
\newblock Shiftaddnet: A hardware-inspired deep network.
\newblock {\em NeurIPS}, 33, 2020.

\bibitem{zhang2018residual}
Yulun Zhang, Yapeng Tian, Yu Kong, Bineng Zhong, and Yun Fu.
\newblock Residual dense network for image super-resolution.
\newblock In {\em CVPR}, pages 2472--2481, 2018.

\end{thebibliography}
}

\end{document}